\documentclass[runningheads]{llncs}
\usepackage[utf8]{inputenc}
\usepackage[x11names,dvipsnames,table]{xcolor} 
\usepackage{mathtools}
\usepackage{amsmath,amsfonts,amssymb}
\usepackage{enumitem}
\usepackage{graphicx,subfigure}
\usepackage{wasysym}
\usepackage{algpseudocode,algorithm}
\usepackage{capt-of}
\usepackage{cleveref}
\usepackage{amsfonts}
\usepackage{array}
\usepackage{sidecap}

\newtheorem{obs}{Observation}
\newcommand{\polylog}{\mathrm{polylog}}

\algtext*{EndFor}
\algtext*{EndIf}
\begin{document}
	\title{Computing The Packedness of Curves}
	\author{Sepideh Aghamolaei\inst{1} \and Vahideh Keikha\inst{2}\and
	Mohammad Ghodsi\inst{1} \and
		Ali Mohades\inst{4} }
	\authorrunning{Aghamolaei et al.}
	\institute{Department of Computer Engineering, Sharif University of Technology,\and The Czech Academy of Sciences, Institute of Computer Science,\and
		Department of Mathematics and Computer Science, Amirkabir University of Technology,
		\\ \email{aghamolaei@ce.sharif.edu, va.keikha@gmail.com, ghodsi@sharif.edu, mohades@aut.ac.ir}}

	\maketitle

\begin{abstract}
A polygonal curve $P$ with $n$ vertices is $c$-packed, if the sum of the lengths of the parts of the edges of the curve that are inside any disk of radius $r$ is at most $cr$, for any $r>0$. Similarly, the concept of $c$-packedness can be defined for any scaling of a given shape.

Assuming $L$ is the diameter of $P$ and $\delta$ is the minimum distance between points on disjoint edges of $P$, we show the approximation factor of the existing $O(\frac{\log (L/\delta)}{\epsilon}n^3)$ time algorithm is $1+\epsilon$-approximation algorithm. The massively parallel versions of these algorithms run in $O(\log (L/\delta))$ rounds. We improve the existing $O((\frac{n}{\epsilon^3})^{\frac 4 3}\polylog \frac n \epsilon)$ time $(6+\epsilon)$-approximation algorithm by providing a $(4+\epsilon)$-approximation $O(n(\log^2 n)(\log^2 \frac{1}{\epsilon})+\frac{n}{\epsilon})$ time algorithm, and the existing $O(n^2)$ time $2$-approximation algorithm improving the existing $O(n^2\log n)$ time $2$-approximation algorithm.

Our exact $c$-packedness algorithm takes $O(n^5)$ time, which is the first exact algorithm for disks. We show using $\alpha$-fat shapes instead of disks adds a factor $\alpha^2$ to the approximation.

We also give a data-structure for computing the curve-length inside query disks. It has $O(n^6\log n)$ construction time, uses $O(n^6)$ space, and has query time $O(\log n+k)$, where $k$ is the number of intersected segments with the query shape.
We also give a massively parallel algorithm for relative $c$-packedness with $O(1)$ rounds.
\end{abstract}

\section{Introduction}

In 2012, Driemel et al.~\cite{driemel2012approximating} introduced the $c$-packedness property for a curve as having curve-length at most $cr$ inside any disk of radius $r$ (\Cref{def:cpackedness}), but did not give an algorithm for computing, approximating, or deciding it.
Let $P=\{P_1,\ldots,P_n\}$ be a polygonal curve, and let $\overline{P_iP_{i+1}}, \forall i=1,\ldots, n-1$ denote the $i$-th edge of $P$.
		\begin{definition}[$c$-packed curve~\cite{driemel2012approximating}]\label{def:cpackedness}
		Let $P$ be a polygonal curve. 
		If the length of $P$ within any disk of radius $r$ is upper bounded by $cr$, $P$ is $c$-packed. 
	\end{definition}
%
The packedness of a disk $D$ of radius $r$ with respect to a curve $P$ is denoted by $\gamma (D,P)$ and defined as the length of $P$ inside $D$, divided by $r$. Formally,
				$\gamma (D,P)=\frac {\sum_{i=1}^{n-1} |\overline{P_iP_{i+1}} \cap D|}{r}.$
\Cref{fig:defs} shows an example with $\gamma(D,P)\approx 4$.
\begin{SCfigure}
\includegraphics[scale=0.6]{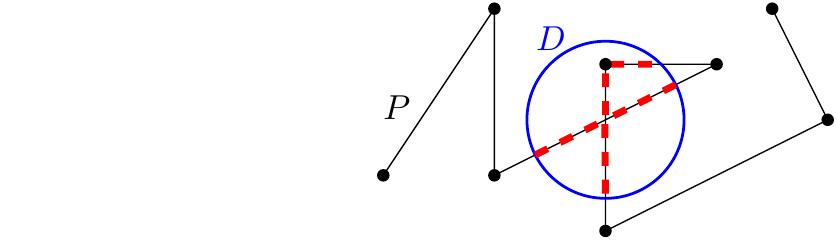}
\caption{A polygonal path $P$ and a disk $D$ with packedness about $4$.}
\label{fig:defs}
\end{SCfigure}

After introducing this concept, fairly good approximation algorithms for $c$-packed curves have been presented. Driemel et al.~\cite{driemel2012approximating} have shown that if two curves are $c$-packed, a $(1+\epsilon)$-approximation of their Fr{\'e}chet distance can be achieved in $O(\frac {cn}{\epsilon}+cn \log n)$ time.  Bringmann and K\"unnemann further improved this running time to $O(\frac {cn}{\sqrt{\epsilon}}\log^2 (\frac 1 \epsilon)+cn \log n)$~\cite{bringmann2015improved}.
Several other studies~\cite{chen2011approximate,driemel2013jaywalking,driemel2018probabilistic,har2014frechet,aghamolaei2020windowing,gudmundsson2015fast} are only justified when the input curves are $c$-packed.

For axis-aligned squares instead of disks, $c$-packedness was discussed under the name relative-length~\cite{gudmundsson2013algorithms} and they gave an exact $O(n^3)$ algorithm for it.
Computing the minimum $c$ for which a curve is $c$-packed for disks was discussed by Aghamolaei et al.~\cite{aghamolaei2020windowing} who formalized the problem (\Cref{def:problem}) and gave an algorithm for approximating the minimum $c$.

\begin{definition}[Minimum $c$ for $c$-packedness of a curve~\cite{aghamolaei2020windowing}]\label{def:problem}
		For a given polygonal curve $P$ with $n$ vertices, the objective is to compute the minimum value $c$ for which $P$ is $c$-packed, in other words
		$ \min_{\substack{\forall \text{disks } D,\\\gamma(D,P)\le cr}} c. $
\end{definition}

Aghamolaei et al.~\cite{aghamolaei2020windowing} introduced a data structure called aggregated query diagram (AQD) for computing the length of the curve inside all translations of a query shape. The authors designed a $(1+\epsilon)$-approximation algorithm for disks by approximating a disk with a regular polygon, and a $(1+\epsilon)$-approximation algorithm for computing the $c$-packedness within any convex polygon (without rotation), with running times $O(\frac{\log (L/\delta)}{\epsilon}n^3)$ and $O(\frac{\log (L/\delta)}{\epsilon^{5/2}}n^3\log n)$, respectively.

  \begin{definition}[Aggregated Query Diagram (AQD)~\cite{aghamolaei2020windowing}]
For a set $P$ of $n$ line segments and a polygonal shape $X$, a partitioning $c_1,\ldots,c_k$ of the plane with a set $O_i\subset S$ at each $c_i$, is called an aggregated query diagram if for any point $q\in c_i$, the query shape with representative point $q$ intersects the subset $O_i$ of segments, and $k$ is minimized. 
\end{definition}

\Cref{fig:aqd} is an example of an AQD for the curve and disk of \Cref{fig:defs}.

The number of the cells in the AQD of a convex shape $X$ of complexity $m$ is $O(n^2m^2)$, and its total complexity is $O(n^3m^2)$. Also, it takes $O(n^3m^2 \log(mn))$ time to construct the corresponding AQD of $X$.  

Later in~\cite{gudmundsson2020approximating}, the authors provided 2-approximation algorithm with running time $O(dn^2\log n)$ in $\mathbb{R}^d$, and a $(6+\epsilon)$-approximation algorithm based on WSPD~\cite{wspd} with running time $O((\frac{n}{\epsilon^3})^{\frac 4 3}\polylog \frac n \epsilon)$ in $\mathbb{R}^2$, for axis-aligned squares with the maximum packedness centered at the vertices of the curve.

\begin{figure}[t]
	\centering
	\includegraphics[scale=0.6]{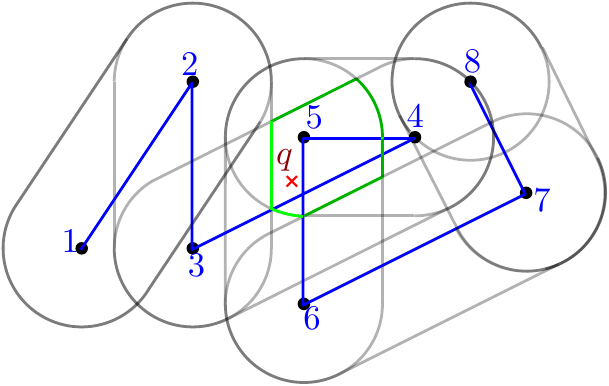}
	\caption{The AQD of a polygonal path for a disk query. The green region is the cell containing the query point $q$ (in red), so, it intersects with $\{\overline{P_3P_4},\overline{P_4P_5},\overline{P_5P_6}\}$.}
	\label{fig:aqd}
\end{figure}

A generalization of the relative $c$-packedness is the maximum of  the packedness of queries centered at a given point-set $S$.
For a set $S$ of points and a curve $P$, if for any disk of centered at a point of $S$, the length of the path inside the disk divided by its radius is at most $c$, the curve is an {\em $S$-relative $c$-packed} curve.

$c$-packedness can be defined for any shape $X$ of fixed orientation and its scalings. Let $P=\{P_1,\ldots,P_n\}$ be a polygonal curve, and $X$ be a shape with the smallest enclosing disk of diameter $2r$. If the length of $P$ within any translation of $X$ is upper bounded by $cr$, then $P$ is $c$-packed.

In the massively parallel computations (MPC)~\cite{beame2017communication} data is distributed among a sublinear number of machines $O(n^{1-\eta})$ each with sublinear memory $O(n^{\eta})$, that process it during $O(\textrm{poly}(\frac{1}{\eta}))$ parallel rounds and communicate with each other after each round. The parallel version of batched queries~\cite{agarwal1990partitioning} are simultaneous geometric queries in MPC. Some examples in MPC are $2$-sided range queries~\cite{aghamolaei2018geometric} and length queries~\cite{aghamolaei2020windowing}. We improve the simultaneous length queries algorithm to take $O(\textrm{poly}(\frac{1}{\eta}))$ rounds instead of $O(\log \frac{L}{\delta})$.

\paragraph{Contributions}

\begin{itemize}
	\item We show the lower-bound on the vertex-relative $c$-packedness for $c$-packedness is $2$, which matches the upper-bound of the existing algorithms.
	
	\item We give a massively parallel algorithm for vertex-relative $c$-packedness.
	
	\item We give an $O(n^5)$ time exact algorithm for the minimum $c$-packedness. 
	
	\item We give a $G(n/\epsilon^2,n)$ time $(4+\epsilon)$-approximation algorithm, improving the $(6+\epsilon)$-approximation algorithm, and a $O(n^2)$ time $2$-approximation algorithm improving the existing $O(n^2\log n)$ time $2$-approximation algorithm.
	
	\item We give a data-structure for length queries using disks of arbitrary size, called Hierarchical Aggregated Query (HAQ) data-structure. It can be constructed in $O(n^6\log n)$ time, using $O(n^6)$ space, and has query time $O(\log n+k)$.
\end{itemize}
\paragraph{A Summary of The Results on Minimum $c$-packedness}
\Cref{table:results} summarize the results on minimum $c$-packedness and vertex-relative minimum $c$-packedness.
\begin{table}[h]
\centering
\begin{tabular}{|p{3cm}|c|c|c|c|}
\hline
Shape & Approx. & Time & Reference & Relative $c$-Packed\\
\hline\hline
circle & $1+\epsilon$ & $O(\frac{\log (L/\delta)}{\epsilon}n^3)$ & Section~\ref{sec:SL} & -\\
convex polygon (constant complexity) & $2+\epsilon$ & $O(\frac{\log (L/\delta)}{\epsilon}n^3\log n)$ & \cite{aghamolaei2020windowing}& -\\
circle& $2+\epsilon$ & $O(\frac{\log (L/\delta)}{\epsilon^{5/2}}n^3\log \frac{n}{\epsilon})$ & \cite{aghamolaei2020windowing}& -\\
square & exact & $O(n^3)$ & \cite{gudmundsson2013algorithms}& -\\
$d$-cube & $2$ & $O(dn^2 \log n)$ &  \cite{gudmundsson2020approximating}& yes\\
square & $6+\epsilon$ & $O((\frac{n}{\epsilon^3})^{\frac 4 3}\polylog \frac n \epsilon)$ & \cite{gudmundsson2020approximating}& yes\\
\hline
circle & exact & $O(n^5)$ & \Cref{thm:events}& -\\
$\alpha$-fat shapes & $\alpha^2\beta$ & $T_c(n)$ & \Cref{theorem:fat}& -\\
circle & exact & $O(n^5)$ & \Cref{thm:events}& -\\
\hline
circle, square & $2$ & $O(n^2)$ & \Cref{theorem:improved} & yes\\
circle, square & $4+\epsilon$ & $O(n(\log^2 n)(\log^2 \frac{1}{\epsilon})+\frac{n}{\epsilon})$ & \Cref{thm:wspd} & yes\\
\hline
\end{tabular}
\caption{A summary of results on minimum $c$-packedness of curves. $L$ is the length of the curve. $\delta$ is the minimum distance between two points from disjoint edges. $T_c(n)$ is the time complexity of a $\beta$-approximation for $c$-packedness using circles.}
\label{table:results}

\begin{tabular}{|p{3cm}|c|c|c|c|p{2cm}|}
\hline
Shape & Approx. & Rounds & Memory & Reference & Constraint\\
\hline\hline
convex polygon (constant complexity) & $2(1+\epsilon)$ & $O(\log (\frac{L}{\delta}))$ & $O(\frac{nk}{\epsilon})$ & \cite{aghamolaei2020windowing}& $k$ x-monotone curves\\
circle& $2(1+\epsilon)$ & $O(\log (\frac{L}{\delta}))$ & $O(\frac{nk}{\epsilon^{5/2}})$ & \cite{aghamolaei2020windowing}& $k$ x-monotone curves\\
\hline
circle& $1+\epsilon$ & $O(\log (\frac{L}{\delta}))$ & $O(\frac{nk}{\epsilon^{5/2}})$ & \cite{aghamolaei2020windowing}, Section~\ref{sec:SL}& $k$ x-monotone curves\\
circle& 2 & $O(1)$ & $O(n)$ & \Cref{theorem:parallel}&-\\
\hline
\end{tabular}
\caption{Minimum $c$-packedness in the MPC model. Memory means the total memory of all machines. Polynomial factors of $1/\eta$ are ignored.}
\label{table:parallel}
\vspace{-0.8cm}
\end{table}
\section{Extending Previous Works}

\subsection{The Black-Box Versions of The Algorithms of~\cite{aghamolaei2020windowing}}\label{sec:general}
Here, we discuss the subroutines of the algorithms in~\cite{aghamolaei2020windowing}  and their analysis, some of which can be substituted by similar ones in a black-box manner.

\paragraph{Point-Location among Curves and Length Query}
For a  set of $n$ segments and a shape $X$, the length query asks
for finding the total sum of the lengths of the parts of the segments that lie inside $X$~\cite{aghamolaei2020windowing}.

The time complexity of the length query using the algorithm of~\cite{aghamolaei2020windowing} is $O(\log n+k)$ for convex polygons and congruent disks and using the algorithm of~\cite{aghamolaei2020point} is $O(\log n+k)$ for other closed convex shapes defined by polynomials of constant degree. In~\cite{aghamolaei2020windowing}, the length query becomes a point-location among a set of curves, by taking the Minkowski sum of the input polygonal curve with the query shape, so the time complexity of length query is the time complexity of point-location in AQD. Other point-location algorithms in a set of curves also exist~\cite{hanniel2000two,har2001online}, but they either provide no theoretical guarantees or their theoretical guarantee is worse than the algorithm that we used.
Alternatively, length query can be computed via a linear scan, which takes $O(n)$.

\paragraph{Popular Places}
For a set of segments, a polygon $X$, the goal of length-based popular places~\cite{aghamolaei2020windowing} is to find a translation of $X$ such that the length of the part of the intersecting segment inside $X$ is maximized.
The original definition asks for the maximum number of curves that intersect a translation of $Q$~\cite{benkert2010finding}. Other variations have also been solved~\cite{aghamolaei2020windowing}.

\paragraph{AQD for disks} In~\cite{aghamolaei2020windowing}, first, the Minkowski sum of the curve and the query shape is computed, then the results are divided into a set of line segments and a set of circles, finally the intersection of the line-segments is computed via line-sweeping and the intersection of $n$ disks of radius $r$ with the segments centered at the vertices of the path is computed by checking every (disk,edge) pair to see if they intersect. The number of intersections between the segments is $O(n^2)$, so, computing the length of the curve inside each of them by scanning the edges of the curve takes $O(n)$, resulting in $O(n^3)$ time.

\paragraph{Algorithm}

The algorithm has two steps:
\begin{enumerate}
\item In the first step we compute a range for the candidate of the solution.
\begin{itemize}
\item[Continuous] In the presented algorithm in~\cite{aghamolaei2020windowing}, the range of radius $r$ is continues and varies from $\delta$ to $L$, where $\delta$ is the closest pair distance between the points on two non-intersecting edges.
The lower bound ($\delta$) is the distance between the closest pair from non-intersecting edges. By checking the distances between all (vertex,vertex) and (vertex,edge) pairs, this can be computed in $O(n^2)$ time.

\item[Discrete] In Algorithm of \Cref{sec:exactalg}, the range of search is discrete. This range is determined by $O(n^4)$ events.
\end{itemize}

\item
In the second step, we check the packedness of the candidates from the first step. We can use any of the existing algorithms. 
\begin{itemize}
\item[Continuous] For each value of $r$ (from a discrete or continues range), we solve the popular places problem on $P$. This was used in~\cite{aghamolaei2020windowing}.
\item[Discrete] For each event, we compute the length query.
\end{itemize}
\end{enumerate}

\paragraph{Analysis}
\begin{itemize}
\item[Continuous]
Let $T(n)$ denote the preprocessing time and $S(n)$ be the size of an AQD data-structure in the arrangement of a set of curves. Note that the arrangement of the Minkowski sum of the shapes with the input curve has complexity $O(n^2)$, assuming the query shape has constant complexity.

The time complexity of the length query algorithm is
$
O(\log_{1+\epsilon}R) (T(n^2)+\textrm{O}(S(n^2))),
$
which is the number of candidate values checked to get a $(1+\epsilon)$-approximation of the radius, times the time complexity of computing the AQD and checking its cells for the solution.
\item[Discrete]
Let $Q(n)$ be the query time of a length query.
The running time of the discrete algorithm is the number of events times the time complexity of computing the length query for each event, which is $E\times Q(n)$.

For a set of disks, $E=O(n^4)$ (see~\Cref{sec:exactalg}), and using the naive algorithm for length query $Q(n)=O(n)$, resulting in the time complexity $O(n^5)$.
\end{itemize}
\paragraph{Higher Dimensions}
Using the Minkowski sum in $\mathbb{R}^d$~\cite{deza2018diameter}, 
 the AQD of a curve in $\mathbb{R}^d$ can be built by considering the intersection of a set of polyhedras. Then, the cell with the maximum length gives the solution.

\paragraph{Massively Parallel Computation}
The time complexity of minimum $c$-packedness~\cite{aghamolaei2020windowing} in MPC rounds is the logarithm of the range of the candidate radii checked by the algorithm which is $O(\log R)=O(\log \frac{L}{\delta})$, since the popular places problem can be solved in polynomial time after partitioning among the machines, if the number of intersections of the curve with a vertical line is at most $k$~\cite{aghamolaei2020windowing}.

\subsection{Approximating $c$-packedness of Convex Polygons and Disks}

\paragraph{Polygons}
In~\cite{aghamolaei2020windowing}, first the AQD is computed and the set of intersected cells are stored at each cell. Then, they the maximum length is computed over all the cells. This algorithm takes $O(\frac{\log (L/\delta)}{\epsilon}n^3\log n)$ time. 

\paragraph{Disks}
	In~\cite{aghamolaei2020windowing}, for a disk $X$,
	the length of the curve within $X$ is
	bounded by the length of the curve inside the largest inscribed regular $k$-gon in $X$	and the length of the curve inside the smallest circumscribing regular $k$-gon of $X$. For $k=\frac {\pi} {\sqrt{\epsilon}}$,
	  the radius of the circumscribing circle of the regular $k$-gon would be at most $1+\epsilon$ times the radius of the inscribed circle. This gives a $(1+\epsilon)$-approximation.
	
The running time of the $2(1+\epsilon)$-approximation of~\cite{aghamolaei2020windowing}(\Cref{alg:cpackedness} in this paper, Algorithm~3 in~\cite{aghamolaei2020windowing}) is $O(\frac{\log (L/\delta)}{\epsilon}n^3\log n)$, which is the time cost of the search on the range of the radius, multiplied by the time complexity of solving the maximum length query problem for each radius . \Cref{alg:cpackedness} uses the improved range discussed in~\Cref{sec:general} instead of the original bounds in~\cite{aghamolaei2020windowing}.
\begin{algorithm}[t]
	\caption{Approximating $c$-Packedness for Disks~\cite{aghamolaei2020windowing} (Slightly Modified)}
	\label{alg:cpackedness}
	\begin{algorithmic}[1]
		\Require{A curve $P=\{P_1,\ldots,P_n\}$, a constant $\epsilon>0$}
		\Ensure{The minimum $c$ for which $P$ is $c$-packed}
		\State{$c=0, \epsilon\gets \epsilon/2$}
		\State{$L=\max_{p,q\in (\cup_{i=1}^n\{P_i\})} \lVert p-q\rVert$, $\delta=\min_{\overline{P_iP_{i+1}}\ni p} \min_{\overline{P_jP_{j+1}}\ni q,\overline{P_iP_{i+1}} \cap \overline{P_jP_{j+1}} =\emptyset} \lVert p-q\rVert$}
		\For{$i=0,\ldots,\lceil\log_{1+\epsilon} L/\delta\rceil$}
		\State{$r=\delta (1+\epsilon)^i$}
		\State{$Q=$ a regular polygon with $\lceil \frac{\pi}{\sqrt{\epsilon}}\rceil$ vertices}
		\State{cost= the maximum length query in the AQD of $Q$.}
		\State{$c= \max\{c,\frac{cost}{r}(1+\epsilon)\}$}
		\EndFor
		\\ \Return{$c$}
	\end{algorithmic}
\end{algorithm}
The complexity of the shape $Q$ is $\lceil \frac {\pi} {\sqrt{\epsilon}}\rceil$.
So, computing the maximum length query takes $O((\frac{n}{\sqrt{\epsilon}})^3\log \frac{n}{\sqrt{\epsilon}})$ time, which takes the maximum of the length queries in the AQD cells.
The running time of the algorithm is 
$O(\frac{\log (L/\delta)}{\epsilon^{5/2}}n^3\log \frac{n}{\epsilon})$.

In the proof of Theorem~2 in~\cite{aghamolaei2020windowing}, it is stated that for a specific set of segments that a query object $Q$ intersects, the local maximums of the total sum of the lengths in different placements of $Q$ (or equivalently, the events in the current study) are at the intersections of a set of parabolas with each other or the domain boundary. Such functions are the intersections of segments and the boundaries of $AQD$ cells, so, these intersections are only a function of the parts of the edges that define query shapes with a different set of intersected edges (\Cref{lemma:event}).
We assess the maximum values that are determined by these intersection points to find the $c$-packedness of $Q$.

\begin{lemma}[the proof of Theorem~2 in~\cite{aghamolaei2020windowing}]\label{lemma:event}
For a polygonal curve $P$, one event for each continuous set of translations of a query $Q$ that intersects with the same set of edges is enough to compute the query with the maximum length.
\end{lemma} 

\subsection{$S$-Relative $c$-Packedness}

\paragraph{Exact $S$-Relative Packedness}
Assume a set $S$ of size $m$ and a query shape whose intersection with a segment can be computed in $O(1)$ time are given. The number of query shapes in the definition of $s$-relative packedness is $O(nm)$. Computing the length query by intersecting all the edges of the curve with the query shape takes $O(n)$ time. So, the exact value can be computed in $O(mn^2)$.

\paragraph{Approximating $c$-Packedness with Vertex-Relative $c$-Packedness}

Lemma 5 in \cite{aghamolaei2020windowing} states that
if a curve has length at most $cr$ within all the disks of radius $r$ centered at a point of $P$, then $P$ is $2c$-packed. The proof is done via the observation that any disk of radius $r$ that intersects $P$, can be covered by a disk of radius $2r$ centered at a point of $P$, where its center is in the intersection of two disks; and the length of the curve inside the disk of smaller radius is at most as much as the length of the curve inside the larger disk.

When $S$ is the set of vertices of the curve, the $S$-relative $c$-packedness is called the vertex-relative $c$-packedness.
The $2$-approximation algorithm in~\cite{gudmundsson2020approximating} gives an exact $c$ for vertex-relative $c$-packedness using squares and hypercubes since all possible disks centered at vertex with another vertex on its boundary have been considered. Their algorithm has time complexity $O(n^2\log n)$.

\begin{obs}\label{thm:contradiction}
Vertex-relative $c$-packedness is at least twice $c$-packedness, and the packedness is at least $2$; since any disk centered at an edge has packedness $2$, and the vertex-relative packedness of a single edge is $1$.
\end{obs} 
\subsection{Line Sweeping for Minimum Packedness via Continuous Search}
\label{sec:SL}

	When the ratio $L/\delta=O(1)$, the continuous search algorithm outperforms the polynomial ones. So, we discuss this problem as well.
	
	The sweeping events for the case where the center of the disk lies on the curve are when the disk intersects with an edge, or stops intersecting with it. The number of such events is $O(n^2)$, so sweeping over them takes $O(n^2\log n)$ time (similar to the construction of the AQD) and computing the length takes another $O(n)$ time, in the worst-case.

Using the continuous search of~\cite{aghamolaei2020windowing} for approximating the minimum $c$ for which a given curve is $c$-packed, we sweep the curve with a disk of that radius.

\section{A Polynomial-Time Exact Algorithm for $c$-Packedness} \label{sec:exactalg}
Here, we show that disks of maximal size and optimal $c$-packedness can be computed by checking $O(n^4)$ candidate disks.

\Cref{lemma:event}(restated from~\cite{aghamolaei2020windowing}) shows that determining the set of intersected segments (edges of the curve) by a candidate disk of each size is enough to find it exactly.
In~\Cref{lemma:event2}, we generalize that lemma.
\begin{lemma}\label{lemma:event2}
For a polygonal curve $P$, one event for each continuous set of affine transformations of a query disk $Q$ that intersects with the same set of edges $E$, given that the same subset $F$ of $E$ is used to define $Q$, is enough to compute the query with the maximum length.
\end{lemma}
\begin{proof}
A circle can be defined by the two endpoints of its diameter ($|F|=2$), or three points on its boundary ($|F|=3$).
Let $R_1$ be the radius of the smallest disk with the same sets $E$ and $F$, and let $R_2$ be the radius of the largest such disk.
Let $D(c,r)$ be the disk with the same intersection set $E$ and defined by points from $F$ and of radius $r$ centered at a point $c$, and assume $R_1(c)=\min_r D(c,r)$ and $R_2(c)=\max_r D(c,r)$.
Since the curve is polygonal, using the chord-length formula, the rate of increase of the length of a chord with distance $h$ from $c$ in terms of $r$ is
$
\frac{\mathrm{d}}{\mathrm{d}r} \sqrt{r^2-4h^2}= \frac{r}{\sqrt{r^2-4h^2}}=\frac{r}{l},
$
where $l$ is the length of the chord of the disk of radius $r$ containing $s$, for $r\in [R_1(c),R_2(c)]$. Let $m_s=\frac{1}{l}$.
So, the packedness of $D(c,r)$ is
$
\frac{L(c)+\sum_{s\in E} m_s(r+\Delta r-R_1(c))}{r+\Delta r}=\sum_{s\in E} m_s+\frac{L(c)-\sum_{s\in E} m_s R_1(c)}{r+\Delta r},
$
where $L(c)$ is the length of the curve inside $D(c,R_1(c))$.
This function has its maximum in its domain at radius $R_2(c)$, which is $\sum_{s\in E} m_s+\frac{L(c)-\sum_{s\in E} m_s R_1(c)}{R_2(c)}$.
Taking the maximum packedness of disks $D(c,r)$ with the same sets $E$ and $F$ over different points $c$, gives:
$
\max_c \sum_{s\in E} m_s+\frac{L(c)-\sum_{s\in E} m_s R_1(c)}{R_2(c)}= \sum_{s\in E} m_s+\max_c\frac{L(c)-\sum_{s\in E} m_s R_1(c)}{R_2(c)}.
$
Based on the definition of $L(c)$, at least one of the segments has length $0$ inside $D(c,R_1)$, so $L(c)-\sum_{s\in E} m_s R_1(c)\le 0$. So, the maximum happens when $R_2(c)$ is maximized, which happens for the disk of radius $R_2$ with intersection set $E$ and defined by $F$.
\end{proof}
\paragraph{Computing the disk described in \Cref{lemma:event2}}
If smallest enclosing circle of the endpoints of the segments in $F$ contains one point from each segment in $E$, we are done; otherwise, compute the smallest disk $D_s$ intersecting at least one point of $F$, and call this set of points $F_s$. $D_s$ can be computed by enumerating all $\binom{6}{3}$ or $\binom{6}{2}$ possible cases. Then, find the closest point of each segment in $E$ to the center of $D_s$ and call this set $C$. The smallest enclosing disk of $C\cup F_s$ is the solution and it can be computed in $O(|E|)$ time.

\begin{obs}\label{obs:circle}
Any circle containing a point from each of the three segments $s_u,s_v,s_w$ contains the closest point of those segments to each other.
\end{obs}
So, we focus on finding the sets described in~\Cref{lemma:event2}.
Let $C(u,v,w)$ denote the smallest circle containing points $u,v$, and $w$.
\begin{algorithm}[h]
\caption{Enumerating The Events of Exact Minimum $c$-Packedness}
\label{alg:events}
\begin{algorithmic}[1]
\For{each $i,j\in 1,\ldots,n$}
\State{$p_i,p_j=$ the closest points of $\overline{P_iP_{i+1}}$ and $\overline{P_jP_{j+1}}$ to each other.}
\State{$dp[0]=1$}
\For{$k=1,\ldots,n$}
\State{$p_k=$ the point $p$ on segment $s_k$ with the smallest $C(p_i,p_j,p)$.}
\State{$dp[k]=\sum_{\substack{0\leq z<k,\\ p_k\in C(p_i,p_j,p_z)}} dp[z]$}
\If{$\forall z<k, p_k \notin C(p_i,p_j,p_z)$}
\State{$dp[k]=dp[k]+1$}
\EndIf
\EndFor
\EndFor
\end{algorithmic}
\end{algorithm}
\begin{lemma}\label{lemma:quad}
The recurrence relation of \Cref{alg:events} finds $O(n^2)$ events for each pair of segments $s_i$ and $s_j$.
\end{lemma}
\begin{proof}
Using induction on the number of edges checked so far. At step $k\leq n$, the set of edges of the curve that can be used to create the events are $s_1,\ldots,s_k$.
The circle determined by $p_i$ and $p_j$ as the endpoints of its diameter is the base case.
Based on \Cref{obs:circle}, since $p_i,p_j$ and $p_k$ are known, there are two cases:
\begin{itemize}
\item $C(p_i,p_j,p_k)$ determines a circle that was not given by any subset of $s_1,\ldots,s_{k-1}$, i.e. it is a new event.
\item $\exists k': C(p_i,p_j,p_k')$ determines the circle intersecting $s_k$, i.e. there is no circle that intersect $s_i,s_j$ and $s_k$ but not any of $s_z, z<k$.
\end{itemize}
This is all the possibilities of the segments being intersected, which is what the recurrence in the algorithm computes.
\end{proof}

Similarly, the other endpoint of the segments can be considered for removing the segment from the set of intersected segments, and the number of events is still $O(n^2)$, since the number of involved endpoints in computing the events has changed from $n$ to $2n$.

\begin{theorem} \label{thm:events}
The number of events is $O(n^4)$.
\end{theorem}
\begin{proof}
Choosing two segments $s_i$ and $s_j$ has $\binom{n}{2}=O(n^2)$ possibilities, and using \Cref{lemma:quad}, there are $O(n^2)$ events for each of them. This is $O(n^4)$. Each event is counted at most via $\binom{3}{2}=3$ subroutines, which is $O(1)$.
\end{proof}

For any of the computed events, sets $E$ and $F$ are constructed, and the largest disk described in \Cref{lemma:event2} is computed, which takes $O(n)$ time per event, and the algorithm runs in $O(n^5)$ time.
Assuming the packedness of each event is computed at its construction time, and only the maximum is stored, the space complexity of the algorithm is $O(n)$.
\subsection{Packedness using Fat Shapes}
A shape enclosed inside two concentric disks of radii $\rho$ and $\alpha\rho$ is called $\alpha$-fat~\cite{agarwal1995computing}. Approximating an $\alpha$-fat shape with the disk of radius $\rho$ in this definition, gives an $\alpha^2$-approximation for minimum $c$-packedness:
\begin{theorem}\label{theorem:fat}
The packedness of the smaller disk in the definition of an $\alpha$-fat shape is an $\alpha^2$-approximation for the packedness of the shape.
\end{theorem}
\begin{proof}
Let $P$ be a $c$-packed curve, and $D$ be the smaller disk of the optimal translation of query $Q$. Let $l$ be the length of $P$ inside $Q$, and $r$ be the radius of $Q$. Then, $l$ is at least as much as the length of the curve inside $D$, which is $c\rho$, and the radius of $Q$ is at most $\alpha \rho$. So, the approximation factor is $\alpha^2$:
$
c\rho \le \gamma(Q,P) \le c\alpha\rho,\;\rho \le r \le \alpha \rho \Rightarrow \frac{c}{\alpha} \le \frac{l}{r} \le c\alpha.
$
\end{proof} \vspace{-0.4cm}
\subsection{Vertex-Relative $c$-Packedness of Disks}\label{theorem:improved}
In \Cref{alg:cpackedness2}, we use the corollary of~\Cref{lemma:event2} (\Cref{cor:lemma2}) to show $c$-packedness and vertex-relative $c$-packedness are within factor $2$ of each other.
\begin{corollary}\label{cor:lemma2}
Vertex-relative $c$-packedness is at most twice $c$-packedness.
Using \Cref{lemma:event2}, there is an optimal disk defined by the vertices of the curve, so, there is a disk centered at one of those vertices with twice the radius that covers it.
\end{corollary}
We modify the $2$-approximation algorithm of~\cite{gudmundsson2020approximating} for squares and hypercubes to use \Cref{cor:lemma2}, which improves the running time to $O(n^2)$.

\begin{algorithm}[t]
	\caption{Relative $c$-Packedness for Disks}
	\label{alg:cpackedness2}
	\begin{algorithmic}[1]
		\Require{A curve $P=\{P_1,\ldots,P_n\}$}
		\Ensure{The minimum $c$ for which $P$ is vertex-relative $c$-packed}
		\State{$L_x=$ the edges of $P$ sorted on $x$, $L_y=$ the edges of $P$ sorted on $y$.}
		\For{$i=1,\ldots,n$}
			\State{$L=$ merge $L_x$ and $L_y$ based on their distance to $p$}
			\State{Traverse $L$ and update the length $\ell$ and radius $r$ at each point.}
			\State{$c=\max(c,\frac{\ell}{r})$}
		\EndFor
	\end{algorithmic}
\end{algorithm}
The list $L$ is sorted based on the distances from $p$, since increasing the radius is a monotone movement in the directions of the axes, so the intersection of the disk with a line through its center and parallel to each axis is also monotone. A similar idea was used for squares in~\cite{gudmundsson2013algorithms}, but the algorithm was different.

Updating the lengths is possible by updating $\sum_{s\in E} m_s$ from \Cref{lemma:event2}, so it takes $O(1)$ time. So, the algorithm takes $O(n^2)$ time.

\subsection{An MPC Algorithm for Vertex-Relative $c$-Packedness of Disks}\label{theorem:parallel}
In \Cref{alg:parallel}, the parallel version of \Cref{alg:cpackedness2} is given. 
\begin{algorithm}[h]
	\caption{Parallel Relative $c$-Packedness for Disks}
	\label{alg:parallel}
	\begin{algorithmic}[1]
		\Require{A curve $P=\{P_1,\ldots,P_n\}$}
		\Ensure{The minimum $c$ for which $P$ is vertex-relative $c$-packed}
		\State{Sort the edges to build $L_x$ and $L_y$.}
		\State{Run a parallel-prefix for each point $P_i$ to compute the packedness of disks centered at $P_i$ and one of $P_j, j=1,\ldots,n$ on their boundaries.}
		\State{Run a parallel semi-group to take the maximum.}
	\end{algorithmic}
\end{algorithm}
The algorithm uses two sortings, $n$ simultaneous parallel-prefix computations and a semi-group, each of which take $O(\frac{1}{\eta})$ rounds~\cite{goodrich2011sorting,goodrich2011sorting2}, so the algorithm also takes $O(\frac{1}{\eta})$ rounds. \vspace{-0.4cm}
\subsection{An Algorithm Using WSPD}\label{thm:wspd}
A pair of sets $(A,B)$ of points are $s$-separated~\cite{wspd}, if disks $C_A$ and $C_B$ of radius $\rho$ containing the bounding boxes of $A$ and $B$ fits inside them and the distance between $C_A$ and $C_B$ is at least $s\rho$.
A well-separated pair decomposition (WSPD)~\cite{wspd} of a point-set $S$ is a set of $s$-separated pairs $\{(A_i,B_i)\}_{i=1}^{m}$ from $S$, such that for any points $p,q\in S$, there is exactly one index $i$ such that $p\in A_i$ and $q\in B_i$, or $p\in B_i$ and $q\in B_i$. The size of a WSPD is $m=O(s^2 n)$.
A disk defined by a WSPD pair $(A_i,B_i)$ has its center in $A_i$ and covers $B_i$.
\begin{lemma}\label{lemma:wspd}
The relative packedness of the disks defined by WSPD pairs is a $(2+\epsilon)$-approximation for the relative packedness and there are $O(n/\epsilon^2)$ disks.
\end{lemma}
\begin{proof}
For a disk $D^*$ with optimal relative packedness centered at $p$ and with $q$ on its boundary, there is a $s$-separated pair $(A_i,B_i)$ with $p\in A_i, q\in B_i$ or $p\in B_i,q\in A_i$.
Based on the properties of WSPD, the diameter of $D^*$ is $d\geq s\rho$.
The smallest disk $D$ centered at $p$ containing $A_i\cup B_i$ has radius at most $4\rho+d$.
Let $\ell$ be the length of the curve inside $D^*$. Then, the packedness of $D$ is:
$
\gamma(D)\geq \frac{\ell}{4\rho+d} \geq \frac{\ell}{4d/s+d}\geq \frac{\ell}{d/2}\frac{1}{8/s+2}\geq \gamma(D^*)\frac{2s+8}{s}=\gamma(D^*) (2+\frac{8}{s}).
$
For $s=\frac{8}{\epsilon}$, this is a $(2+\epsilon)$-approximation for squares and using square length-queries.
\end{proof}


\paragraph{Approximate Length Queries for Squares}
 By a line-sweeping~\cite{de1997computational} over the edges of the curve and the queries, we first break each segment between each endpoint and the closest query side, resulting in at most $3n$ segments. Segments with both endpoints on queries (set $W$) and ones with at least one of their original endpoints (set $M$) are solved using different methods.
 
We partition the segments $W$ into a set of groups $[\frac{(i-1) \pi}{\sqrt{\epsilon}},\frac{i \pi}{\sqrt{\epsilon}}]$, for $i=1,\ldots, \frac{2 \pi}{\sqrt{\epsilon}}$, in which the difference between different slopes in a group is at most $\frac{\pi}{\sqrt{\epsilon}}$.
The idea of discretizing $\theta$ already existed in spanners like $\Theta$-graphs and Yao-graphs~\cite{wspd}, $\epsilon$-kernels~\cite{agarwal2004approximating} and range-queries~\cite{aghamolaei2018geometric}.
Each segment with inclination angle $\theta$ with endpoints on a strip of width $\Delta r$ has length $\frac {\Delta r}{\cos \theta}$. So we can get a $(1+\epsilon)$-approximation of the total length in any strip.
To handle square queries, in addition to $x,y,\Delta \theta$, we need another dimension $\Delta r$ which discretizes the perimeter $\epsilon \Delta r$.
Since we group the last two dimensions, their trees have $O(\frac{1}{\sqrt{\epsilon}})$ elements.
This is similar to type B queries in~\cite{gudmundsson2020approximating}. 

We preprocess $W$ into a stabbing windowing data-structure on $(x,y,\Delta \theta, \Delta r)$ for counting queries~\cite{de1997computational}. To answer length queries, aggregate on the last two dimensions, by computing the weighted sum using $m_s$ described in~\Cref{lemma:event2}.

 For segments $M$, i.e. segments with at least one endpoint from the original input, use an axis-parallel windowing query~\cite{de1997computational}, compute the sum, and divide it by the cell length. Note that if a segment of $M$ is intersected by a square query the segment falls completely inside the query.
The sum of these two values gives a $(1+\epsilon)$-approximation of the length queries in $O(n(\log^2 n)(\log^2 \frac{1}{\sqrt{\epsilon}})+\frac{n}{\epsilon})$ time.
\section{A Data Structure for Exact Length Queries}
%
The results of~\cite{aghamolaei2020windowing} allows querying translations of the query shape $Q$. We also allow querying all scalings with scale factor at least one (enlargements) of $Q$.
Using the events for the candidate values of the radii of the disks with maximum packedness, we build a set of AQDs for each disk of that radius. We call this data-structure {\em hierarchical aggregated query (HAQ)} data-structure (See \Cref{alg:ds}).

\begin{algorithm}[t]
\caption{Building The Data-Structure HAQ for Disks}
\label{alg:ds}
\begin{algorithmic}[1]
\Require{A polygonal curve $P$}
\State{$E=$ the set of radii of the events for the $c$-packedness of $P$ for disk queries}
\State{Build a balanced search tree $T$ on $E$}
\For{$r\in E$}
\State{Build an AQD for radius $r$ and store it at the node with number $r$ in $T$.}
\EndFor
\\ \Return{The augmented tree $T$.}
\end{algorithmic}
\end{algorithm}

During query time, the AQD built for the smallest radius that is greater than or equal to the radius of the query shape is used (See \Cref{alg:query}).
\begin{algorithm}[t]
\caption{Length Query using HAQ for Disks}
\label{alg:query}
\begin{algorithmic}[1]
\Require{An HAQ data-structure, a query disk $Q$}
\State{$v=$ Search HAQ to find the smallest radius that is greater than or equal to the radius of $Q$.}
\State{In the AQD of node $v$, find the center of $Q$ and compute the length.}
\end{algorithmic}
\end{algorithm} 

\begin{theorem}
The HAQ for disks (\Cref{alg:ds,alg:query}) can be constructed in $O(n^6\log n)$ time using $O(n^6)$ space, and the query time is $O(\log n+k)$, where $k$ is the number of intersected segments with the query shape.
\end{theorem}
\begin{proof}
The size of an AQD of a polygonal curve with disk queries is $O(n^2)$. Based on \Cref{thm:events}, there are $O(n^4)$ different radii. So, the size of the data-structure is $O(n^6)$. The construction time of AQD is $O(n^2\log n)$, so the construction time of HAQ is $O(n^6\log n)$.

The query time consists of the time required for finding the value of $r$, which is the first level of tree and takes $O(\log n)$ time. Then, a query is made to the AQD for radius $r$, which takes $O(\log n+k)$, since the intersection events do not change between the radii in HAQ and the set of intersected segments remains $k$.

\Cref{alg:query} computes an exact solution, since the query shape falls inside the query shape used in the construction of the AQD.
\end{proof}

Given the set of events for a shape other than a disk, by replacing the AQD for disks with AQD for polygons, the HAQ data-structure extends to convex polygons. The complexities change according to the time and space complexity of the data-structure.
%
%
%

\bibliographystyle{abbrv}
\bibliography{refs}

\end{document}